\documentclass[12pt,reqno]{amsart}
\usepackage{amsthm}
\usepackage{amsmath}
\usepackage{latexsym}
\usepackage{amsfonts}
\usepackage{amssymb}
\usepackage{color}
\usepackage{bbm,dsfont}
\usepackage{graphicx}
\usepackage{subfigure}
\usepackage{enumerate}
\usepackage{bbding}

\newtheorem{proposition}{Proposition}
\newtheorem{theorem}{Theorem}
\newtheorem{lemma}{Lemma}

\theoremstyle{definition}

\newtheorem{remark}{Remark}
\newtheorem{example}{Example}

\newcommand{\real}{\mathbb R} 
\newcommand{\complex}{\mathbb C} 
\newcommand{\nat}{\mathbb N} 
\newcommand{\half}{\tfrac{1}{2}} 
\newcommand{\mo}[1]{\left| #1 \right|} 

\newcommand{\hi}{\mathcal{H}} 
\newcommand{\lh}{\mathcal{L(H)}} 
\newcommand{\lsh}{\mathcal{L}_s(\hi)} 
\newcommand{\sh}{\mathcal{S(H)}} 
\newcommand{\eh}{\mathcal{E(H)}} 
\newcommand{\ip}[2]{\left\langle\,#1\,|\,#2\,\right\rangle} 
\newcommand{\ket}[1]{|#1\rangle} 
\newcommand{\bra}[1]{\langle#1|} 
\newcommand{\tr}[1]{\mathrm{tr}\left[#1\right]} 
\newcommand{\ran}{\textrm{ran}\,} 
\newcommand{\id}{\mathbbm{1}} 

\newcommand{\A}{\mathsf{A}}
\newcommand{\B}{\mathsf{B}}
\newcommand{\C}{\mathsf{C}}
\newcommand{\D}{\mathsf{D}}
\newcommand{\E}{\mathsf{E}}
\newcommand{\F}{\mathsf{F}}

\newcommand{\T}{\mathsf{T}}
\newcommand{\N}{\mathsf{N}}

\newcommand{\eid}{\mathfrak{1}}
\newcommand{\enul}{\mathfrak{0}}

\newcommand{\state}{\mathcal{S}} 
\newcommand{\effect}{\mathcal{E}} 
\newcommand{\reffect}{\tilde{\mathcal{E}}} 
\newcommand{\meter}{\mathcal{M}} 
\newcommand{\rmeter}{\tilde{\mathcal{M}}} 
\newcommand{\noise}{\mathcal{N}} 
\newcommand{\trivial}{\mathcal{T}} 
\newcommand{\simu}[1]{\mathfrak{sim}(#1)} 
\newcommand{\lmin}{\lambda_{\min}} 
\newcommand{\lmax}{\lambda_{\max}} 

\begin{document}

\title[]{Operational restrictions in general probabilistic theories}

\author[]{Sergey N. Filippov}
\email{sergey.filippov@phystech.edu}
\address{Steklov Mathematical Institute of Russian Academy of Sciences, Moscow  119991, Russia}
\address{Valiev Institute of Physics and Technology of Russian Academy of Sciences, Moscow  117218, Russia}
\address{Moscow Institute of Physics and Technology, Dolgoprudny, Moscow Region 141700, Russia}

\author[]{Stan Gudder}
\email{sgudder@du.edu}
\address{Department of Mathematics, University of Denver, Denver, CO 80208, U.S.A.}

\author[]{Teiko Heinosaari}
\email{teiko.heinosaari@utu.fi}
\address{QTF Centre of Excellence, Department of Physics and Astronomy, University of Turku, Turku 20014, Finland}

\author[]{Leevi Lepp\"{a}j\"{a}rvi}
\email{leille@utu.fi}
\address{QTF Centre of Excellence, Department of Physics and Astronomy, University of Turku, Turku 20014, Finland}



\begin{abstract}
The formalism of general probabilistic theories provides a universal
paradigm that is suitable for describing various physical systems
including classical and quantum ones as particular cases. Contrary to the usual no-restriction hypothesis, the set
of accessible meters within a given theory can be limited for
different reasons, and this raises a question of what restrictions
on meters are operationally relevant. 
We argue that all operational restrictions must be closed under simulation, where
the simulation scheme involves mixing and classical post-processing
of meters. We distinguish three classes of such operational
restrictions: restrictions on meters originating from restrictions
on effects; restrictions on meters that do not restrict the
set of effects in any way; and all other restrictions. We
fully characterize the first class of restrictions and discuss its connection to
convex effect subalgebras. We show that the restrictions belonging to the second class can impose severe physical limitations despite the fact that all
effects are accessible, which takes place, e.g., in the
unambiguous discrimination of pure quantum states via effectively
dichotomic meters. 
We further demonstrate that there are physically meaningful restrictions that fall into the third class.
The presented study of operational restrictions provides a better understanding on how accessible measurements
modify general probabilistic theories and quantum theory in particular.
\end{abstract}

\maketitle

\section{Introduction}

The framework of general probabilistic theories (GPTs) provides an abstract setting for possible physical theories based on
operational principles. Containing not only quantum and classical
theories but also countless toy theories in between and beyond,
GPTs give us means to study well-known properties of quantum
theory (such as measurement incompatibility \cite{BuHeScSt13}, steering
\cite{StBu14,Banik15}, entanglement \cite{AuLaPaPl19} and no-information-without-disturbance \cite{HeLePl19}) in a more general setting.
This allows us to formulate
and examine these properties in different theories, quantify them
and even compare different theories to each other based on how
these properties behave within them. Many properties that were
thought to be special features of quantum theory have actually been shown to be general among all non-classical probabilistic
theories, the no-broadcasting theorem being perhaps the most well known example \cite{BaBaLeWi07}.

One of the most long-standing motivations has been to provide a set of physical principles, formulated in the GPT framework, that would lead to an axiomatic derivation of quantum theory.
In recent years, followed by the success of quantum information theory, there has been a new boom of such efforts and many
information-theoretic axioms have been proposed from which the
quantum theory has been successfully derived \cite{QTIFF16}.
In addition to a full physical axiomatization, one can focus on some specific
property of interest and study it independently of the underlying
theory with the aim of finding something meaningful on the nature
of the property itself.

GPTs are based on operational notions of states, effects,
measurements, transformations, and composite systems so that by
specifying them one fixes the theory. The most important
operational principle for describing the state space $\state$ of
the theory is the statistical mixing of states which then leads to
$\state$ being a convex subset of a real vector space. As the most
simple type of measurements, the effects are then taken to be
affine functionals $e: \state \to [0,1]$ that give probabilities
on states so that $e(s)$ can then be interpreted as the
probability of observing the effect $e$ when the system is
measured in state $s \in \state$. The affinity of effects is a
result of the basic statistical correspondence between states and
measurements. A meter that corresponds to a measurement device can
then be described as a normalized collection of effects. A meter
provides a generalization of the positive operator-valued measure
(POVM) in quantum theory.

The assumption of taking all mathematically valid affine
functionals that give probabilities on states as physical effects
of the theory has been coined as the no-restriction hypothesis
\cite{ChDaPe10}. The no-restriction hypothesis is satisfied in
both classical and quantum theories, so it is usually accepted in
other theories too for the purpose of mathematical convenience. If
the no-restriction hypothesis is assumed, then the (single-system)
theory is completely determined by the state space alone. However,
as it has been pointed out, e.g., in \cite{JaLa13}, the
no-restriction hypothesis has no operational grounds. In fact, it
is possible to provide different kinds of consistent restrictions
on the set of effects that then give rise to new models and have
consequences even on the way the composite systems could be formed
\cite{JaLa13}. Other works beyond the no-restriction hypothesis
are, e.g., \cite{ChDaPe11, BaMuUd14, Wilce19}.

Interestingly, in the recent work \cite{SaGuAcNa18} it was shown
that the no-restriction hypothesis plays a significant role in the
correlations that can be achieved within quantum theory. In
particular, it was shown that a set of correlations that is close
to the set of quantum correlations, called the almost-quantum
correlations, violate the no-restriction hypothesis. This means
that no GPT with the no-restriction hypothesis is able to
reproduce the almost-quantum correlations. Therefore, the
no-restriction hypothesis may be a crucial part of singling out
the quantum correlations from other non-signalling theories.

Even if we restrict to the quantum theory, there is also a
practical motivation to investigate restrictions on meters and their consequences.
For example,
conventional measurement schemes for superconducting qubits and
polarized photons perform dichotomic measurements in the
computational basis or the rotated computational
basis~\cite{ClDeGiMaSc10}. Measurements with more than two outcomes
cannot be directly implemented for such two-level systems. To
obtain more than two outcomes one usually resorts to mixing and
post-processing dichotomic observables instead.
Therefore, only effectively dichotomic observables are available in
conventional quantum experimental setups with no entanglement
between the system and an ancilla. The use of the ancilla enables
one to perform measurements with a greater number of
outcomes, the number of measurement outcomes depending on the
dimension of the ancillary system.
Moreover, even the dichotomic measurements are never perfectly
projective~\cite{ClDeGiMaSc10,NaPo14}, which imposes a restriction on the noise content of accessible meters.
Another example of practical restrictions is that the effects for fermionic systems are not arbitrary and must satisfy the parity superselection rule~\cite{AmFi17}.

In the current work we consider restrictions not only at the level
of effects but also on the level of meters, and we show that the
previously studied effect restrictions are not enough to capture
all operationally valid restrictions. We propose an operational
condition that any restriction on meters should satisfy, namely
the \emph{simulation closedness} criterion. In accordance with the operational interpretation of GPTs, for a given set of meters there are two classical operations one can always implement that will lead to some outcome statistics differing from those of any other meter that may be used. In particular, similarly to mixing states, one can choose
to mix meters, and after the measurement it is possible to
post-process the obtained outcomes. The scheme consisting of both
mixing and post-processing of meters, called the measurement
simulability, has been previously studied in \cite{OsGuWiAc17,
GuBaCuAc17, FiHeLe18}. Our operational condition of simulation
closedness for meters then states that given a set of allowed
meters as a restriction, also all meters that can be obtained by
the simulation scheme from the allowed ones should be included in
the restriction as well. A violation of this condition would mean
that some classical procedure consisting of mixing and
post-processing of outcomes is not allowed, and that would
therefore be a weird and unphysical restriction.

We show that the introduced operational restrictions can be
divided into three disjoint classes: (R1) restrictions on meters
that are dictated by the restrictions on effects, (R2)
restrictions on meters that do not restrict the effects in any
way, and (R3) restrictions on meters that cannot be reproduced by
any restriction solely on effects, but nevertheless restrict the
set of effects as well. We demonstrate these restrictions in
quantum theory.

Our investigation is organized as follows. A brief overview of the
relevant concepts is given in Sec.~\ref{sec:prelim}. In
Sec.~\ref{sec:restrictions} we introduce the classification of
operational restrictions into three disjoint classes (R1)--(R3).
In Sec.~\ref{sec:cea} we characterize those effect restrictions that give simulation closed restrictions of type (R1) and examine convex effect algebras and their
subalgebras and see how they are related to these restrictions. Effectively $n$-tomic theories are presented as a class of
restrictions of type (R2) and they are examined in
Sec.~\ref{sec:ntomic}.
In Sec.~\ref{sec:R3} we give examples of restrictions that belong to (R3).
Finally, in Sec.~\ref{sec:conclusions} we summarize our investigation.

\section{Preliminaries}\label{sec:prelim}

\subsection{States, effects, meters}\label{sec:states}
We start by recalling the ordered vector space formulation of GPTs
(for more details see, e.g., \cite{BaWi11}). The state space
$\state$ of a GPT is a compact convex subset of a
finite-dimensional real vector space $V$. Whereas compactness and
the finite-dimensionality of the state space are merely technical
assumptions, the convexity follows from the possible statistical
mixing of the states: if we can prepare our system in states $s_1
\in \state$ or $s_2 \in \state$, by fixing some $p \in [0,1]$ we
can choose to use state $s_1$ with probability $p$ and state $s_2$
with probability $1-p$ in each round of the experiment so that $p
s_1 + (1-p) s_2$ must be a valid state in $\state$.

If $\dim(\mathrm{aff}(\state)) =d$, then $V$ can be chosen to be
$(d+1)$-dimensional and $\state$ forms a compact base for a closed
generating proper cone $V_+$ \footnote{A subset $C \subset V$ of a
vector space $V$ is a (convex) cone if $C +C \subseteq C$ and
$\alpha C \subseteq C$ for every $\alpha\in\real^+$. Furthermore,
$C$ is a proper cone if $C \cap (- C) = \{0\}$ and generating if
$C -C = V$. A subset $B \subset C$ is a base of $C$ if for every
$x \in C \setminus \{0\}$ there exists unique $\beta > 0$ and $b
\in B$ such that $x = \beta b$.}. The cone $V_+$ defines a partial
order in $V$ in the usual way; we denote $v \leq w$ (or $v
\leq_{V_+} w$ if we want to explicitly write the cone to avoid
confusion) if $w-v \in V_+$. Thus, $V_+$ consists of all of the
positive elements induced by this order. As a base of $V_+$, the
state space $\state$ can be expressed in terms of a strictly
positive functional $u \in V^*$ as
\begin{equation}
\state = \{ s \in V \, | \, s \geq 0,  \ u(s)=1 \}.
\end{equation}

The effect space $\effect(\state)$ consists of affine functionals $e: \state \to [0,1]$ giving probabilities on states: we interpret $e(s)$ as the probability that the effect $e$ is observed when the system is measured in state $s \in \state$. Affinity of effects is a result of them respecting the basic statistical correspondence of states and effects:
\begin{equation}
e(p s_1 +(1-p)s_2) = p e(s_1) +(1-p) e(s_2)
\end{equation}
for all $p \in [0,1]$, $s_1,s_2 \in \state$ and $e \in \effect(\state)$.

In the ordered vector space formulation we can express the effect
space as $\effect(\state) = V^*_+ \cap (u- V^*_+)$, where $V^*_+$
is the (closed generating proper) positive dual cone\footnote{Dual
cone $C^*\subset V^*$ of a cone $C \subset V$ consists of positive
linear functionals on $C$, i.e., $C^* = \{ f \in V^* \, | \, f(x)
\geq 0 \ \forall x \in C\}$. } of $V_+$ in the dual space $V^*$
and $u$ is the unit effect in $V^*_+$. Explicitly,
\begin{equation}
\effect(\state) = \{ e \in V^* \, | \, o \leq e \leq u \},
\end{equation}
where $o$ is the zero effect that gives value $0$ for every state and where the partial order is now the dual order induced by the dual cone $V^*_+$.

An effect $f \in \effect(\state) \subset V^*$, $f \neq o$, is called \emph{indecomposable} if whenever a decomposition $f= f_1 +f_2$ of $f$ into a sum of some other nonzero effects $f_1, f_2 \in \effect(\state)$ implies that $f= \alpha_1 f_1 = \alpha_2 f_2$ for some $\alpha_1, \alpha_2 >0$. The indecomposable effects are precisely the effects lying on the extreme rays of the dual cone $V^*_+$.
It was shown in \cite{KiNuIm10} that every effect can be decomposed into a sum of some indecomposable effects and that indecomposable extreme effects exist in all GPTs.

For an effect $f \in \effect(\state)$, we denote by $\lmin(f)$ and $\lmax(f)$ its smallest and largest values on $\state$, i.e., $\lmin(f) = \inf_{s \in \state} f(s)$ and $\lmax(f) = \sup_{s \in \state} f(s)$.
We note that these are attained because $\state$ is compact and $f$ is continuous.

A meter $\A$ with an $n$ outcomes is a mapping $\A: x \to \A_x$
from an outcome set $\Omega_\A = \{1, \ldots, n\} \subset \nat$ to
the set of effects $\effect(\state)$ such that the normalization
condition $\sum_{x \in \Omega_\A} \A_x = u$ is satisfied. Thus,
the set $\Omega_\A$ includes all the possible outcomes of the
experiment where the meter $\A$ is used, the normalization
condition guarantees that some outcome is registered, and
$\A_x(s)$ can be then interpreted as the probability that outcome
$x \in \Omega_\A$ was observed when the systems was in the state
$s \in \state$ and meter $\A$ was used to measure the system. We
denote the set of meters on $\state$ as $\meter(\state)$, or
simply as $\meter$ if the state space is understood from the
context.

For the purpose of this work it is worth noting that when we presented the usual definition of the effect space, no further restrictions on its elements was given. This means that all mathematically
valid functionals (i.e., affine functionals that give
probabilities on states) are also considered to be valid physical
effects in the theory. This assumption is commonly called the
no-restriction hypothesis. In this work we give operationally
justifiable restrictions that we pose on the unrestricted set of
effects/meters but unless otherwise stated, the underlying set of
effects of the theory is taken to be unrestricted.

\begin{example}
In finite-dimensional quantum theory, the state space $\state(\hi)$ consists of positive trace-1 operators on a finite-dimensional Hilbert space $\hi$, i.e.,
\begin{equation}
\state(\hi) := \{ \varrho \in \lsh \, | \, \varrho \geq O, \ \tr{\varrho } = 1\},
\end{equation}
where $O$ is the zero operator on $\hi$, $\lsh$ denotes the real vector space of self-adjoint operators on $\hi$ and the order is induced by the cone of positive-semidefinite operators on $\hi$, i.e., $A \geq O$ if and only if $\ip{\varphi}{A \varphi} \geq 0$ for all $\varphi \in \hi$.

The effect space $\effect(\sh)$ can be shown to be isomorphic to the set of selfadjoint operators between the zero operator $O$ and the identity operator $\id$, i.e.,
\begin{equation}
\effect(\sh) \cong \eh := \{ E \in \lsh \, | \, O \leq E \leq \id \},
\end{equation}
where naturally the zero operator $O$ corresponds to the zero functional $o$ and the identity operator $\id$ corresponds to the unit effect $u$.

Each meter on $\sh$ with a finite number of outcomes can be
associated with a positive operator-valued measure (POVM) $\A: x
\to \A(x)$ from a finite outcome set $\Omega_\A$ to the set of
effects $\eh$ such that $\sum_{x \in \Omega_\A} \A(x) = \id$.
\end{example}

\subsection{Simulation of meters}
Given a set of measurement devices (meters) one can always choose
to do some classical manipulations with the measurement data
outputted by the devices. For instance, one can consider if it is
possible to construct some new meters by classically manipulating
the pre-existing meters and their measurement data. These type of
considerations have led to the concept of measurement simulability
and have been studied in \cite{OsGuWiAc17, GuBaCuAc17, FiHeLe18,
CaHeMiTo19}.

By classical manipulations we mean mixing the meters and/or
post-processing the outcomes of the meters: If we have meters
$\B^{(1)}, \ldots, \B^{(m)}$ we can assign to them probabilities
$p_1, \ldots, p_m$ of using the different meters in each round of
the measurement process so that we obtain a mixed meter $\B =
\sum_i p_i \B^{(i)}$. In addition to mixing, we can classically
post-process the measurement outcomes of any $\B^{(i)}$ by
assigning a stochastic post-processing matrix $\nu^{(i)} = \left(
\nu_{xy}^{(i)} \right)_{x \in \Omega_{\B^{(i)}}, y \in
\Omega_{\A^{(i)}}}$ to each $\B^{(i)}$, where $\Omega_{\B^{(i)}}$
is the outcome set of the pre-existing meter $\B^{(i)}$ and
$\Omega_{\A^{(i)}}$ is some other outcome set such that
$\nu^{(i)}_{xy} \geq 0$ and $\sum_{y \in \Omega_{\A^{(i)}}}
\nu^{(i)}_{xy} =1$ for all $x \in \Omega_{\B^{(i)}}$, $y \in
\Omega_{\A^{(i)}}$. We can use $\nu^{(i)}$ to define a new meter
$\A^{(i)} = \nu^{(i)} \circ \B^{(i)}$ with outcome set
$\Omega_\A^{(i)}$ by setting $\A^{(i)}_y = \sum_{x \in
\Omega_{\B^{(i)}}} \nu^{(i)}_{xy} \B^{(i)}_x$ for all $y \in
\Omega_{\A^{(i)}}$. Here the matrix element $\nu_{xy}$ can thus be
interpreted as the transition probability that the outcome $x$ is
mapped into outcome $y$.

By combining both mixing and post-processing we get the simulation scheme which results in a new meter $\A$ defined by
\begin{equation}
\A_y = \sum_{i=1}^m p_i (\nu^{(i)} \circ \B^{(i)})_y = \sum_{i=1}^m \sum_{x \in \Omega_\B} p_i \nu^{(i)}_{xy} \B^{(i)}_x,
\end{equation}
for all $y \in \Omega_\A$, where we have set all the outcome sets $\Omega_{\B^{(i)}}$ equal, and denoted the resulting outcome set $\Omega_\B$, by adding zero outcomes if needed, and similarly for $\Omega_\A$.

We denote the set of meters obtained from the meters $\B^{(1)}, \ldots, \B^{(m)}$ by this simulation scheme with some probability distribution $(p_i)_i$ and post-processings $\nu^{(i)}$ by $\simu{\{ \B^{(1)}, \ldots, \B^{(m)} \}}$. If we have a (possibly infinite) set of meters $\mathcal{B}$, we denote by $\simu{\mathcal{B}}$ the set of meters that can be simulated by using some finite subset of $\mathcal{B}$, and call meters in $\mathcal{B}$ as simulators. One can show that $\simu{\mathcal{B}}$ is closed both under post-processing and mixing, i.e., $\simu{\mathcal{B}}$ is convex and $\nu \circ \B \in \simu{\mathcal{B}}$ for any post-processing $\nu$ and meter $\B \in \simu{\mathcal{B}}$.

Being considered as a mapping on the power set $2^\meter$, the
simulation map $\simu{\cdot }$ can be shown to be a closure
operator so that it satisfies the following three properties for
all subsets $\mathcal{B}, \mathcal{C} \subseteq \meter$:
\begin{enumerate}
\item[(SIM1)] $\mathcal{B} \subseteq \simu{\mathcal{B}}$
\item[(SIM2)] $\simu{\simu{\mathcal{B}}} = \simu{\mathcal{B}}$
\item[(SIM3)] $\mathcal{B} \subseteq \mathcal{C} \Rightarrow \simu{\mathcal{B}} \subseteq \simu{\mathcal{C}}$
\end{enumerate}

We call a subset of meters $\mathcal{B}$ \emph{simulation closed} if the equality holds in (SIM1), i.e., $\simu{\mathcal{B}} = \mathcal{B}$. By the property (SIM2) we see that $\simu{\mathcal{B}}$ is simulation closed for any $\mathcal{B} \subseteq \meter$. Simulation closed sets have some basic properties. In particular, if $\mathcal{B}_i$, $i\in I$, are simulation closed sets, then also $\bigcap_{i \in I} \mathcal{B}_i $ is simulation closed.

\section{Three types of operational restrictions}\label{sec:restrictions}

In this work, by a restriction we will mean that the allowed or
possible meters belong to a subset $\rmeter\subset\meter$. We require the following condition for all restrictions:
\begin{enumerate}
\item[(SC)] simulation closedness: $\simu{\rmeter}=\rmeter$
\end{enumerate}
As has been explained above, given a set of meters, we can always
choose to mix them or post-process their outcomes so that any
meter that can be simulated this way from the pre-existing meters
should always be a feasible meter as well. Given a non-simulation
closed restriction $\rmeter$, we can make it simulation closed by
taking its simulation closure $\simu{\rmeter}$.

 We note that simulation closedness implies that all trivial meters are always included in the restriction as they can be post-processed from any meter. By a trivial meter we mean a meter $\T$ of the form $\T_x = p_x u$ for all $x \in \Omega_\T$ for some probability distribution $(p_x)_x$ on $\Omega_\T$ so that it does not give any information about the input state. In practice, trivial meters can always be implemented just by ignoring the input state and choosing an outcome according to some fixed probability distribution.

In the following, by a restriction we mean a choice $\rmeter\subset\meter$ that satisfies the condition (SC). We recall that the range of a meter $\A$ can be expressed as $\ran(\A) = \{ \sum_{y \in \tilde{\Omega}} \A_y \, | \, \tilde{\Omega} \subseteq \Omega_\A \}$. We use the following notation.
\begin{itemize}
\item For a subset $\rmeter\subset\meter$, we denote by $\effect_{\rmeter}$ the set of all $e\in\effect$ such that $e\in\ran(\A)$ for some $\A\in\rmeter$.
\end{itemize}
Given a restriction $\rmeter$, the set of possible effects is then $\effect_{\rmeter}$.

We can also consider restrictions on meters induced by some restriction on effects. For this, we also use the following notation:
\begin{itemize}
\item For a subset $\reffect \subset \effect$, we denote by $\meter_{\reffect}$ the set of all $\A\in\meter$ such that $\ran(\A) \subset \reffect$.
\end{itemize}

As in \cite{JaLa13}, we impose some consistency conditions for $\reffect$ to generate a restriction $\meter_{\reffect}$:
\begin{itemize}
\item[(E1)] $u \in \reffect$ as it is an essential part of the definition of a meter, and
\item[(E2)] for every $e \in \reffect$, there exists $\A \in \meter_{\reffect}$ such that $e \in \ran(\A)$, i.e., for every physical effect $e \in \reffect$ we must have a way to implement it as a part of some meter.
\end{itemize}
As previously, (SC) is also required to hold for restrictions $\meter_{\reffect}$ given by some effect restriction $\reffect$.

The previous considerations lead to the following classification of measurement restrictions into three disjoint cases. Firstly, we can have
\begin{enumerate}
\item[(R1)] $\rmeter = \meter_{\reffect}$ for some $\reffect\subset\effect$.
\end{enumerate}
In this case the restriction takes place essentially on the level
of effects and the limitations on meters can be seen as a
consequence. We emphasize that the set $\reffect$ must be chosen
specifically so that $\meter_{\reffect}$ is simulation closed. We
will show that a necessary and sufficient condition for an effect
restriction $\reffect$ that satisfies the consistency conditions
(E1) and (E2) to be simulation closed is that $\reffect$ is a
convex subset of $\effect$. In particular, this is the case when
$\reffect$ is a convex subalgebra of $\effect$; this will be
discussed in detail in Section \ref{sec:cea}.

Secondly, we can have
\begin{enumerate}
\item[(R2)] $\effect_{\rmeter}=\effect$ (but $\rmeter \neq
\meter$).
\end{enumerate}
In this case, the restriction does not limit the possible effects but only how they compose into meters.
A restriction satisfying (R2) cannot satisfy (R1), as $\meter_\effect=\meter$ and we are assuming that a restricted set $\rmeter$ is a proper subset of $\meter$.
An important class of restrictions of type (R2) are restrictions to \emph{effectively $n$-tomic} meters \cite{GuBaCuAc17, FiHeLe18} and, in fact, any restriction of type (R2) contains effectively dichotomic meters.
This class of restrictions is described and studied in Sec. \ref{sec:ntomic}.

The third possibility is that the restriction is neither (R1) nor
(R2). This means that
\begin{enumerate}
\item[(R3)] $\effect_{\rmeter}\subset\effect$ and $\rmeter \neq \meter_{\reffect}$ for any $\reffect\subset\effect$.
\end{enumerate}
In this case there are limitations already at the level of effects, but there are also limitations that come visible only at the level of meters. Restrictions of this type will be considered in Sec. \ref{sec:R3}.

Finally, we note that there can also be other operational requirements that one might want to hold depending on the restriction.
One such requirement might be tomographic completeness:
\begin{enumerate}
\item[(TC)] tomographic completeness: $\A(s_1)=\A(s_2) \  \forall \A\in\rmeter  \ \Rightarrow \  s_1=s_2$.
\end{enumerate}
This requirement is relevant, e.g., if one starts from a more
general framework of convex structures and then needs to justify
that the set of states is a convex subset of a real vector space
\cite{Gudder73}. However, in this work we concentrate on (SC) and
we do not study other requirements.

\begin{remark}
In \cite{JaLa13}, in addition to (E1) and (E2), also convexity of $\reffect$ along with two other consistency conditions are required to hold:
\begin{itemize}
\item[(E3)] for any two effects $e, f \in \reffect$ such that $e,f
\in \ran(\A)$ for some physical meter $\A$, we must have $e+f \in
\reffect$, and \item[(E4)] the adjoint $T^*$ of a linear state
transformation $T: \state \to \state$, defined by $[T^*(e)](s) =
e(T(s))$ for all states and effects, must give a valid effect for
all valid effects, i.e., $T^*(e) \in \reffect$ for all $e \in
\reffect$.
\end{itemize}
In particular, one can show that the effect restrictions considered in \cite{JaLa13} induce restrictions on meters that are simulation closed. We see that the condition (E3) is built in the definition of $\meter_{\reffect}$: if $e,f \in \reffect$ such that $e,f \in \ran(\A)$ for some \emph{physical} meter $\A$, then according to our definition of physicality, we must have $\A \in \meter_{\reffect}$ so that in particular $e+f \in \ran(\A) \subset \reffect$.

The point we want to emphasize is that even if we are considering restrictions on meters given by restrictions on effects, we must also consider how our physical effects are connected to our physical meters. In our work this is done by defining $\meter_{\reffect}$ and in \cite{JaLa13} this is addressed by the condition (E3). Thus, the condition (E3) is different in nature to (E1) and (E2) as it is not expressed only in terms of effects but involves also meters. Regarding (E4), we do not consider state transformations in our current work.
\end{remark}

\section{Restriction class (R1) and convex effect algebras}\label{sec:cea}

In this section we provide a characterization of restrictions of
type (R1). We then consider a more special case of convex effect
restrictions, namely the convex effect subalgebras. This type of
restriction has been used, e.g., in \cite{BaMuUd14}. We derive a
representation theorem for convex effect subalgebras and we also
demonstrate that there are physically meaningful (R1) restrictions
that do not have the structure of a  convex effect subalgebra.
The material presented in Subsections  \ref{sec:cea-intro} and \ref{sec:cea-char}
 has some overlap with the recent work \cite{Gudder19} of one of the present authors.
We include this material to make the present investigation self-contained.

\subsection{Characterization of (R1) restrictions}

As was described earlier, we consider restrictions of type (R1) to be induced by a subset $\reffect\subset \effect$ of effects that satisfies the consistency conditions (E1) and (E2) such that $\meter_{\reffect}$ describes the physical, restricted set of meters that is simulation closed.
We start by showing some simple consequences of the consistency conditions (E1) and (E2) which will be seen useful later.

\begin{lemma}\label{lemma:effect-restrictions}
Let $\reffect \subset \effect$ be a restriction on effects such that consistency conditions (E1) and (E2) are satisfied. Then
\begin{itemize}
\item[(a)] $o \in \reffect$,
\item[(b)] for each $e \in \reffect$ also the complement effect $u-e \in \reffect$.
\end{itemize}
\end{lemma}
\begin{proof}
(a) By (E1) and (E2) there exists a meter $\A \in \meter_{\reffect}$ such that $u \in \ran(\A)$ and since $o \in \ran(\B)$ for any meter $\B \in \meter$, we must have from the definition of $\meter_{\reffect}$ that $o \in \ran(\A) \subset \reffect$.

(b) By (E2) for any $e \in \reffect$, there exists a meter $\A \in \meter_{\reffect}$ such that $e \in \ran(\A)$. Since $u-e \in \ran(\A)$, we have from the definition of $\meter_{\reffect}$ that $u-e \in \reffect$.
\end{proof}

We can now give a complete characterization of effect restrictions $\reffect$ that give rise to restrictions of type (R1).

\begin{theorem}\label{thm:R1}
Let $\reffect \subset \effect$ be a restriction on effects such that consistency conditions (E1) and (E2) are satisfied. Then $\meter_{\reffect}$ is simulation closed if and only if $\reffect$ is convex.
\end{theorem}
\begin{proof}
Let first $\meter_{\reffect}$ be simulation closed. If $e, f \in \reffect$ then from (E2) it follows that there exist $\A,\B \in \meter_{\reffect}$ such that $e \in \ran(\A)$ and $f \in \ran(\B)$. In fact, we have that $u-e \in \ran(\A)\subset \reffect $ and $u-f \in \ran(\B)\subset \reffect $ so that  if we define two dichotomic meters $\E$ and $\F$ with effect $e, u-e$ and $f, u-f$ respectively, then $\E, \F \in \meter_{\reffect}$. Now from (SC) it follows that $t \E +(1-t) \F \in \meter_{\reffect}$ for any $t \in [0,1]$ so that $t e +(1-t)f \in \ran(t \E +(1-t) \F) \subset \reffect$. Thus, $\reffect$ is convex.

Let now $\reffect$ be convex. Let $\A \in \simu{\meter_{\reffect}}$ so that there exist meters $\{\B^{(i)}\}_i \subset \meter_{\reffect}$, post-processings $\nu^{(i)}: \Omega_\B \to \Omega_\A$ and a probability distribution $(p_i)_i$ such that $\A = \sum_i p_i (\nu^{(i)} \circ \B^{(i)})$. We need to show that $\A \in \meter_{\reffect}$, i.e., that $\ran(\A) \subset \reffect$. Since $\ran(\A) = \{ \sum_{y \in \tilde{\Omega}} \A_y \, | \, \tilde{\Omega} \subseteq \Omega_\A \}$, we take $\tilde{\Omega} \subseteq \Omega_\A$ and consider the effect
\begin{equation}
\sum_{y \in \tilde{\Omega}} \A_y = \sum_{y \in \tilde{\Omega}} \sum_i \sum_{x \in \Omega_\B} p_i \nu^{(i)}_{xy} \B^{(i)}_x = \sum_i p_i \left[ \sum_{x \in \Omega_\B} \left( \sum_{y \in \tilde{\Omega}} \nu^{(i)}_{xy} \right) \B^{(i)}_x \right].
\end{equation}
Let us denote $\tilde{\nu}^{(i)}_x := \sum_{y \in \tilde{\Omega}} \nu^{(i)}_{xy} \in [0,1]$ so that
\begin{equation}
\sum_{y \in \tilde{\Omega}} \A_y = \sum_i p_i \left( \sum_{x \in \Omega_\B}\tilde{\nu}^{(i)}_{x} \B^{(i)}_x \right).
\end{equation}
From the convexity of $\reffect$ we see that if $\sum_{x \in \Omega_\B}\tilde{\nu}^{(i)}_{x} \B^{(i)}_x \in \reffect$ for all $i$, then $\sum_{y \in \tilde{\Omega}} \A_y \in \reffect$ which would prove the claim. Thus, we will fix $i$ and focus on $\sum_{x \in \Omega_\B}\tilde{\nu}^{(i)}_{x} \B^{(i)}_x$ and show that it is contained in $\reffect$.

Since $\Omega_\B = \{1, \ldots, n\}$ for some $n \in \nat$, we can rename the effects of $\B^{(i)}$ such that $\tilde{\nu}^{(i)}_x \leq \tilde{\nu}^{(i)}_{x'}$ for $x < x'$. If we set $ \tilde{\nu}^{(i)}_0 = 0$, one can confirm that
\begin{equation}\label{eq:telescope-sum}
\sum_{x=1}^n \tilde{\nu}^{(i)}_x \B^{(i)}_x = \sum_{k=1}^n \left[ \left(\tilde{\nu}^{(i)}_{k} - \tilde{\nu}^{(i)}_{k-1} \right) \sum_{x=k}^n \B^{(i)}_x \right].
\end{equation}
One sees that $\sum_{x=k}^n \B^{(i)}_x \in \ran(\B^{(i)}) \subset \reffect$ and that $\tilde{\nu}^{(i)}_{k} - \tilde{\nu}^{(i)}_{k-1} \geq 0$ for all $k \in \{1, \ldots, n\}$. Furthermore, we see that $\sum_{k=1}^n \left(\tilde{\nu}^{(i)}_{k} - \tilde{\nu}^{(i)}_{k-1} \right) = \tilde{\nu}^{(i)}_n \in [0,1]$ so that we can make the RHS of Eq. \eqref{eq:telescope-sum} a convex sum of the terms $\sum_{x=k}^n \B^{(i)}_x \in \ran(\B^{(i)})$ by adding a zero element $(1- \tilde{\nu}^{(i)}_n ) o$ which by Lemma \ref{lemma:effect-restrictions} must be included in $\reffect$.

Hence, $\sum_{x=1}^n \tilde{\nu}^{(i)}_x \B^{(i)}_x$ can be expressed as a convex combination of elements in $\reffect$ so that from the convexity of $\reffect$ is follows that $\sum_{x \in \Omega_\B}\tilde{\nu}^{(i)}_{x} \B^{(i)}_x \in \reffect$ for all $i$.
\end{proof}

\subsection{Convex effect algebras}\label{sec:cea-intro}

We start by recalling the notion of (abstact) convex effect algebra and the operational basis of this mathematical structure.
An effect algebra \cite{FoBe94} is a non-empty set $\effect$ with two distinguished elements $\enul$ and $\eid$ and a partially defined operation $\oplus$ that satisfies the following conditions:
\begin{itemize}
\item[(EA1)] if $e \oplus f$ is defined, then $f \oplus e$ is defined and $e \oplus f =f \oplus e$.
\item[(EA2)] if $e \oplus f$ and $(e \oplus f)\oplus g$ are defined, then $f \oplus g$ and $e \oplus (f\oplus g)$ are defined and $(e \oplus f)\oplus g = e \oplus (f\oplus g)$.
\item[(EA3)] for every $e \in\effect$, there is a unique $e'$ such that $e \oplus e' = \eid$.
\item[(EA4)] if $e \oplus \eid$ is defined, then $e=\enul$.
\end{itemize}

A physical interpretation of an effect algebra is that $\effect$ is a collection of events and the partial operation $\oplus$ describes joining of events. The element $\enul$ corresponds to event that never happens whereas $\eid$ corresponds to event that always happens. An important example of an effect algebra is the collection of all fuzzy sets on some set $X$.
An abstract effect algebra can be seen as a generalization of this structure, including the Hilbert space effect algebra as an important example.
It is clear that the set of all effects in a GPT also forms an effect algebra.

When thinking about the interpretation of an effect algebra as a collection of events, one could come up with some additional properties that would seem reasonable to require as axioms. However, several such properties can be derived from the defining conditions (EA1)--(EA4).
For instance, it can be shown \cite{FoBe94} that $(e')'=e$ and that the \emph{cancellation law} holds: if $e \oplus f = e \oplus g$, then $f=g$.

Let us then consider an effect algebra that describes events that correspond to outcomes, or collections of outcomes, in a measurement device or devices.
An operational interpretation of the partial operation $\oplus$ is that two outcomes are merged into one. Merging two outcomes is an irreversible action; if we are given the newly formed device, we cannot know which effects have been merged. There is, however, a way to split one outcome into two so that merging is a one side inverse to this procedure. This splitting goes as follows. When an outcome related to an effect $e$ occurs, we toss a coin and, depending on the result, either record the outcome as it was, or mark it as a new outcome. We thus obtain two effects, $e_{same}$ and $e_{new}$. Clearly, merging of the outcomes should give the original effect, thus $e_{same} \oplus e_{new} = e$. In this way we have introduced a map $e \mapsto e_{same}$ for every coin tossing probability $\alpha$.

Mathematically speaking, an effect algebra $\effect$ is a \emph{convex effect algebra} \cite{GuPu98} if for every effect $e\in\effect$ and real number $\alpha \in [0,1]$, we can form a new effect, denoted by $\alpha e$ such that the following conditions hold for every $\alpha,\beta\in [0,1]$ and $e,f\in\effect$:
\begin{itemize}
\item[(CEA1)] $\alpha(\beta e) = (\alpha\beta)e$.
\item[(CEA2)] $1 e = e$.
\item[(CEA3)] If $\alpha + \beta \leq 1$, then $\alpha e \oplus \beta e$ is defined and $(\alpha+\beta)e = \alpha e \oplus \beta e$.
\item[(CEA4)] If $e\oplus f$ is defined, then $\alpha e \oplus \alpha f$ is defined and $\alpha(e \oplus f) = \alpha e \oplus \alpha f$.
\end{itemize}
As we have described above, the map $(\alpha,e) \mapsto \alpha e$
can be interpreted as a splitting of $e$ into two effects, $\alpha
e$ and $(1-\alpha) e$. We point out that this mathematical
structure describes the action only at the level of individual
effects, not meters, which allows for other interpretations. We
can, for instance, interpret the action in a way that the residual
effect $(1-\alpha) e$ does not generate a new outcome but is
combined into some already existing outcomes.

It is shown in \cite{GuPu98} that if $\alpha,\beta \in [0,1]$ with $\alpha + \beta \leq 1$, then $\alpha e \oplus \beta f$ is defined for every $e,f\in\effect$.
This further implies that for any $\alpha\in [0,1]$ and any $e,f\in\effect$, the effect sum $\alpha e \oplus (1-\alpha) f$ is defined.
The resulting element is called a \emph{mixture} of $e$ and $f$.
Mixing of effects is therefore a derived notion in convex effect algebras.

\subsection{Characterization of convex effect algebras and subalgebras}\label{sec:cea-char}

 As we have seen earlier, the partial order in an effect algebra is derived from the partially defined effect sum.
 To construct concrete convex effect algebras, we can start from an ordered vector space and use that structure to form an effect algebra.
This construction works as follows.
Let $W$ be a finite dimensional real vector space, and let $C\subset W$ be a proper cone.
For any nonzero $u\in C$, we then denote $[0,u]_C:=\{e\in C: e \leq_C u \}$.
Then, for any $e,f \in [0,u]_C$, the combination $e \oplus f$ is defined if $e+f \leq_C u$, and then $e \oplus f := e+f$.
The set $[0,u]_C$ is a convex subset of $C$ and $0\in C$.
Therefore, $\alpha e \in [0,u]_C$ for any $e \in [0,u]_C$ and $0\leq \alpha \leq 1$.
In this way, $[0,u]_C$ is a concrete convex effect algebra, also called a \emph{linear effect algebra} \cite{GuPuBuBe99}.
The chosen vector $u$ is the identity element in $[0,u]_C$.

When forming linear effect algebras, we typically want $[0,u]_C$ to generate the vector space $W$, which means that $W$ is the linear span of vectors of the form $\alpha e$, where $\alpha\in\real^+$ and $e\in [0,u]_C$.
Due to the following result it is not restrictive to consider this kind of linear effect algebras when we investigate the properties of convex effect algebras.

\begin{theorem}(\cite{GuPu98})
Let $\effect$ be a convex effect algebra.
There exists a real vector space $W$, a cone $C$ and a nonzero element $u\in C$ such that $[0,u]_C$ generates $W$ and $\effect$ is affinely isomorphic to $[0,u]_C$.
\end{theorem}

We remark that this characterization of convex effect algebras shows a natural connection to the GPT framework.
Namely, if one starts from a GPT state space $\state\subset V$ (see Sec. \ref{sec:states}), then $W$ is the dual space $V^*$ and $C$ is the positive dual cone $V^*_+$.
More detailed discussions about this connection are provided in \cite{BeBu97,BuGuPu00}.

As with any algebraic structures, there are natural notions of substructures for effect algebras and convex effect algebras.
Namely, let $\effect$ be an effect algebra.
A nonempty subset $\reffect\subset\effect$ is a \emph{subalgebra} of $\effect$ if the following conditions hold:
\begin{itemize}
\item[(SA1)] $\eid\in\reffect$.
\item[(SA2)] $e'\in\reffect$ for all $e\in\reffect$.
\item[(SA3)] $e\oplus f\in\reffect$ for all $e,f\in\reffect$ such that $e\oplus f$ is defined in $\effect$.
\end{itemize}
If $\effect$ is a convex effect algebra, then a subalgebra $\reffect$ is a \emph{convex subalgebra} of $\effect$ if it satisfies also the following condition:
\begin{itemize}
\item[(SA4)] $\alpha e\in\reffect$ for all $\alpha\in[0,1]$ and
$e\in\reffect$.
\end{itemize}
We note that every convex effect algebra $\effect$ has two trivial convex subalgebras: $\effect$ itself and $\{ \alpha \eid \, | \, \alpha \in [0,1] \}$.
The following result characterizes all convex subalgebras.

\begin{theorem}\label{thm:subalgebra}
Let $V$ be a vector space, $C$ a cone and $\effect=[0,u]_C$ a convex effect algebra generating $V$.
A subset $\reffect\subset\effect$ is a convex subalgebra if and only if $u\in\reffect$ and there exist $e_1,\ldots,e_n\in\effect$ such that
\begin{align}\label{eq:csa-rep}
\reffect & =\mathrm{span}_\real \{ e_1,\ldots,e_n\} \cap \effect  \nonumber \\
& =\{e\in\effect : e=\sum_i r_i e_i , \, r_i\in\real \} \, .
\end{align}
\end{theorem}

\begin{proof}
Let us assume that $\reffect$ is a subset given in
\eqref{eq:csa-rep} by some elements $e_1,\ldots,e_n\in\effect$,
and $u\in\reffect$. It follows that $u=\sum_i \bar{r}_i e_i$ for
some $\bar{r}_i\in\real$. Using this fact, we see that (SA2) is
valid: if $e\in\reffect$ and hence $e=\sum_i r_i e_i$ for some
$r_i\in\real$, then $e'=\sum_i (\bar{r}_i -r_i) e_i \in \reffect$.
It is clear from \eqref{eq:csa-rep} that also (SA3) and (SA4) are
valid.

Let us then assume that $\reffect$ is a convex subalgebra of
$[0,u]_C$. Let $v_1,\dots,v_m$ be a linear basis in $V$. Since
$[0,u]_C$ generates $V$, every $v_i$ can be written as $v_i=c_i^+
v_i^+ - c_i^- v_i^-$ for some $v_i^+,v_i^-\in[0,u]_C$ and
$c_i^+,c_i^- \geq 0$. We denote $e_i=v_i^+$ and $e_{m+i}=v_i^-$
for $i=1,\ldots,m$. Since $\{v_i\}_{i=1}^m$ is a basis,
\eqref{eq:csa-rep} holds.
\end{proof}

Using the same premises, we can also rephrase Theorem \ref{thm:subalgebra} as follows:
A subset $\reffect\subset\effect$ is a convex subalgebra if and only if $\reffect = U \cap \effect$ for some linear subspace $U \subset V$ such that $u \in U$. Thus, convex subalgebras are always determined by some linear subspace that contains the unit effect.
The smallest nontrivial convex subalgebras are generated by $u$ and some other effect $e$.

\subsection{Subalgebras and restrictions}

We are now ready to explain the connection between convex effect algebras and operational restrictions.
In the following $\effect(\state)$ is the set of all effects on a state space $\state$.
The following statement follows from Theorem \ref{thm:R1}.
Here we give a short direct proof as a consequence of Theorem \ref{thm:subalgebra}.

\begin{proposition} \label{proposition-ces-sim}
Let $\reffect$ be a convex subalgebra of $\effect(\state)$. Then
$\meter_{\reffect}$ is simulation closed.
\end{proposition}

\begin{proof}
We need to show that $\simu{\meter_{\reffect}}\subseteq\meter_{\reffect}$.
Let $\A\in\simu{\meter_{\reffect}}$.
Then
\begin{align*}
\A_x = \sum_i p_i \sum_y \nu^{(i)}_{yx} \B^{(i)}_y,
\end{align*}
where $\B^{(i)}\in\meter_{\reffect}$ and hence
$\B^{(i)}_y\in\reffect$. Since $\reffect$ is a convex subalgebra,
by Theorem \ref{thm:subalgebra} it has representation
\eqref{eq:csa-rep} for some $e_1,\ldots,e_m\in\effect(\state)$. It
follows that $\A_x \in\reffect$, and therefore
$\A\in\meter_{\reffect}$.
\end{proof}

In the following we demonstrate with two propositions that there
are restrictions of the type (R1) where the restricted set of
effects does not form a subalgebra.

Let $\A$ be a meter and denote $\rmeter=\simu{\A}$.
The set $\rmeter$ is simulation closed as $\simu{\simu{\A}}=\simu{\A}$.
The restricted set of effects $\effect_{\rmeter}$ is given as
\begin{align}
\effect_{\rmeter}=\{ e \in \effect: e = \sum_x r_x \A_x \, , r_x \in [0,1] \} \, . \label{eq:simA}
\end{align}
The set $\effect_{\rmeter}$ satisfies the conditions (SA1), (SA2),
and (SA4). However, the condition (SA3) is satisfied only for
specific choices of $\A$; this is the content of the second part
of the following proposition.

\begin{proposition}
Let $\A$ be a meter such that $\{\A_1,\ldots,\A_n\}$ is linearly independent and let $\rmeter=\simu{\A}$.
Then
\begin{itemize}
\item[(a)] $\rmeter = \meter_{\effect_{\rmeter}}$, hence $\rmeter$
is a restriction of type (R1). \item[(b)] $\effect_{\rmeter}$ is a
convex subalgebra of $\effect(\state)$ if and only if
$\lmax(\A_x)=1$ for every $x$.
\end{itemize}
\end{proposition}

\begin{proof}
\begin{itemize}
\item[(a)] From the definitions of $\effect_{\rmeter}$ and
$\meter_{\effect_{\rmeter}}$ it follows that $\rmeter \subseteq
\meter_{\effect_{\rmeter}}$ for any $\rmeter$. For the other
direction, let us take $\B \in \rmeter_{\effect_{\rmeter}}$, where
$\effect_{\rmeter}$ is given by Eq. \eqref{eq:simA}. Thus, for
each $y \in \Omega_\B$ there exist $\{r^{(y)}_x\}_{x \in
\Omega_\A} \subset [0,1]$ such that $\B_y = \sum_{x} r^{(y)}_x
\A_x$. From the normalization of $\A$ and $\B$ we see that
\begin{equation}
\sum_{x \in \Omega_\A} \A_x = u = \sum_{y \in \Omega_\B} \B_y = \sum_{x \in \Omega_\A}  \left( \sum_{y \in \Omega_\B} r^{(y)}_x \right) \A_x
\end{equation}
so that from the linear independence of the effects of $\A$ it follows that $\sum_y r^{(y)}_x = 1$ for all $x \in \Omega_\A$. Thus, we can define a post-processing $\nu: \Omega_\A \to \Omega_\B$ by setting $\nu_{xy} = r^{(y)}_x$ for all $x \in \Omega_\A$ and $y \in \Omega_\B$ so that $\B = \nu \circ \A \in \simu{\A} = \rmeter$. Hence, also $\meter_{\effect_{\rmeter}} \subseteq \rmeter$ holds in this case.

\item[(b)] Let us assume that $\lmax(\A_x)=1$ for every $x$.
Suppose that $e,f\in\effect_{\rmeter}$ and $e+f\leq u$.
We have $e=\sum_x \alpha_x \A_x$ and $f=\sum_x \beta_x \A_x$, and thus $e+f=\sum_x (\alpha_x + \beta_x) \A_x$.
For every $x$, fix $s_x\in\state$ such that $\A_x(s_x)=1$.
Then
\begin{equation*}
1 \geq (e+f)(s_x) = \sum_y (\alpha_y + \beta_y) \A_y(s_x) = \alpha_x + \beta_x \, .
\end{equation*}
Therefore, $e+f\in\effect_{\rmeter}$.

Let us then assume that $1>\lmax(\A_1)\equiv\lambda$. Then
$\tfrac{1}{\lambda} \A_1\in\effect(\state)$ and
$\tfrac{1}{\lambda}>1$. Let $0<\mu<1$ and
$\mu<\tfrac{1-\lambda}{\lambda}$. Then
$\mu\A_1\in\effect_{\rmeter}$ and $\A_1 + \mu \A_1 \leq
\tfrac{1}{\lambda} \A_1$, thus $\A_1 + \mu \A_1 \in
\effect(\state)$. But $\A_1+\mu\A_1\notin\effect_{\meter}$ because
otherwise we would have
\begin{equation*}
(1+\mu)\A_1 = \sum_x r_x \A_x
\end{equation*}
and by linear independence $r_1 = 1+\mu>1$, which is a
contradiction.
\end{itemize}
\end{proof}

\begin{proposition} \label{prop:restricting-isomorphism}
If $\reffect \subsetneq \effect(\state)$ is an effect restriction that satisfies (E1) and (E2) such that there exists an affine bijection $\Phi: \effect(\state) \to \reffect$, then $\reffect$ is not a convex subalgebra of $\effect(\state)$ but it nevertheless gives a restriction of the type (R1).
\end{proposition}

\begin{proof}
Since $\reffect = \Phi(\effect(\state))$, $\effect(\state)$ is
convex, and $\Phi$ is convexity preserving, we have that
$\reffect$ is convex. By Theorem \ref{thm:R1}, we have that
$\meter_{\reffect}$ is simulation closed and thus gives a
restriction of type (R1).

To see that $\reffect$ is not a convex subalgebra of $\effect(\state)$, we note that $\dim\left( \mathrm{span} \left(\reffect\right)\right) = \dim\left( \mathrm{span} \left(\effect(\state)\right)\right) = \dim(V)$ because of the bijectivity. However, from Theorem \ref{thm:subalgebra} we see that the only generating convex subalgebra must be the effect algebra $\effect(\state)$ itself: namely, if $\reffect = U \cap \effect(\state)$ for some subspace $U$, then from the previous equality of the dimensions it follows that also $\dim(U)=\dim(V)$ and this can only be the case when $U=V$ so that $\reffect= V \cap \effect(\state) = \effect(\state)$. Since we have that $\reffect$ is a proper subset of the effect algebra we have arrived at a contradiction and $\reffect$ cannot be a convex subalgebra.
\end{proof}

\section{Restriction class (R2) and effectively $n$-tomic meters}\label{sec:ntomic}

For every integer $n\geq 2$, we use the notation $\meter_{n-\mathrm{eff}} = \simu{\meter_n}$, where $\meter_n$ is the set of all meters that have $n$ or less outcomes. We call $\meter_{n-\mathrm{eff}}$ the set of effectively $n$-tomic meters because they can be reduced to meters with $n$ or less outcomes.
The foundational interest to investigate and test these type of restrictions has been discussed in \cite{KlCa16,KlVeCa17,Huetal18}.

It is clear that $\meter_{n-\mathrm{eff}} \subseteq \meter_{n+1-\mathrm{eff}}$.
The set $\meter_{2-\mathrm{eff}}$ contains all dichotomic meters, therefore $\effect_{\rmeter}=\effect$ for any choice $\rmeter = \meter_{n-\mathrm{eff}}$.
We conclude that these restrictions are of the type (R2).
The restriction to effectively dichotomic meters $\meter_{2-\mathrm{eff}}$ is the smallest restriction of the type (R2) in the following sense, and motivates to look this restriction in more details.

\begin{proposition} \label{proposition-2-eff-subset-R2}
Let $\rmeter\subset\meter$ be an operational restriction of the
type (R2). Then $\meter_{2-\mathrm{eff}} \subseteq \rmeter$.
\end{proposition}

\begin{proof}
Since $\effect_{\rmeter}=\effect$, all dichotomic meters are in $\rmeter$.
As $\rmeter$ is simulation closed, it follows that all effectively dichotomic meters are in $\rmeter$.
\end{proof}

Depending on the theory, it can happen that $\meter_{2-\mathrm{eff}} = \meter$ \cite{FiHeLe18}.
The specific nature of these type of restrictions is hence different in different theories.
There are, however, some general properties of $\meter_{2-\mathrm{eff}}$ and $\meter_{n-\mathrm{eff}}$ that are theory independent; in the following we demonstrate some of these features.

All of the following results are related to the minimal and maximal values $\lmin(\A_x)$ and $\lmax(\A_x)$ of the effects of a meter $\A$.
We start by making some simple observations.
For a dichotomic meter we have $\A_1(s)+\A_2(s)=1$ for all states $s$, and hence
\begin{align}\label{eq:maxmin}
\lmax(\A_2)=1-\lmin(\A_1) \, .
\end{align}
It follows that
\begin{align}\label{eq:maxmax}
\lmax(\A_1)+\lmax(\A_2) \geq 1 \, .
\end{align}
Further, equality holds in \eqref{eq:maxmax} if and only if $\A$
is a trivial meter, i.e., if $\A$ is of the form $\A_1 = p u$ and
$\A_2 = (1-p ) u$ for some $p\in [0,1]$. Clearly, if $\A$ is
trivial meter, then $\lmax(\A_1) + \lmax(\A_2) = p +(1-p) = 1$. On
the other hand, if $\lmax(\A_1) + \lmax(\A_2)=1$, then if we
denote by $s_1$ and $s_2$ the states maximizing $\A_1(s)$ and
$\A_2(s)$ respectively, we see that
\begin{equation*}
1 = \A_1(s) + \A_2(s) \leq \A_1(s_1) + \A_2(s) \leq \A_1(s_1) + \A_2(s_2) =1
\end{equation*}
for all $s \in \state$ so that all the inequalities must actually be equalities and particularly from the first inequality we get that $\A_1(s) = \A_1(s_1)=:q$ for all $s \in \state$. Similarly then $\A_2(s) = \A_2(s_2)=1-q$ for all $s \in \state$. Hence, $\A_1 = q u$ and $\A_2 = (1-q)u$ so that $\A$ is trivial.

\begin{proposition}\label{prop:lambda1}
Any effectively dichotomic meter $\A$ can be simulated from dichotomic meters $\B$ that satisfy
$\lmax(\B_1)=\lmax(\B_2)=1$.
\end{proposition}

\begin{proof}
It is enough to show that any dichotomic meter $\A$ can be post-processed from a dichotomic meter $\A'$ with $\lmax(\A'_1)=\lmax(\A'_2)=1$.
A trivial meter can be post-processed from any meter, so we can further assume that $\A$ is non-trivial.
We denote $\alpha=\lmax(\A_1) + \lmax(\A_2) -1>0$ and define
\begin{align*}
\A'_1 = \tfrac{1}{\alpha} \A_1 + \tfrac{\lmax(\A_2) - 1}{\alpha} u \, , \quad  \A'_2 = \tfrac{1}{\alpha} \A_2 + \tfrac{\lmax(\A_1) - 1}{\alpha} u \, .
\end{align*}
We have $\lmin(\A'_1)=\lmin(\A'_2)=0$ and $\A'_1+\A'_2=u$, hence $\A'$ is a meter.
Further, $\lmax(\A'_1)=\lmax(\A'_2)=1$.
Finally, $\A$ is a post-processing of $\A'$.
\end{proof}

An obvious question is: when a meter $\A\in\meter$ with $m>n$ outcomes is effectively $n$-tomic?
In the following we develop some criteria.

\begin{proposition}\label{prop:eigen-gen}
Let $\A$ be a meter.
\begin{enumerate}
\item[(a)] If there exists $y \in \Omega_\A$ such that $\sum_{x\neq y} \lmax(\A_x) \leq 1$, then $\A$ is effectively dichotomic.
\item[(b)] If $\sum_x \lmax(\A_x) > n$, then $\A$ is not effectively $n$-tomic.
\end{enumerate}
\end{proposition}

\begin{proof}
\begin{enumerate}
\item[(a)] This is a direct generalization of Lemma 5 in the Supplemental Material for Ref. \cite{OsGuWiAc17}, where it was shown to hold for POVMs. 

\item[(b)] Let $\A$ be effectively $n$-tomic, i.e., there exist
$n$-outcome meters $\{ \B^{(i)}\}_i$, post-processings $\nu^{(i)}:
\{1, \ldots, n\} \to \Omega_\A$, and a probability distribution
$(p_i)_i$ such that $\A_x = \sum_i p_i \sum_j \nu^{(i)}_{jx}
\B^{(i)}_j$ for all $x \in \Omega_\A$. Now we see that
\begin{align*}
\sum_x \lmax(\A_x) &= \sum_{x} \max_{s \in \state} \A_x(s) \\
&= \sum_{i,x} p_i \max_{s \in \state} \left( \sum_j \nu^{(i)}_{jx} \B^{(i)}_j\right)(s) \\
&= \sum_{i,x} p_i \max_{s \in \state} \left( \sum_j \nu^{(i)}_{jx} \B^{(i)}_j(s)\right) \\
& \leq \sum_{i,x} p_i \sum_j \nu^{(i)}_{jx}  = \sum_i p_i \left[ \sum_j \left( \sum_x \nu^{(i)}_{jx} \right) \right] = n.
\end{align*}
\end{enumerate}
\end{proof}

The previous result already shows some tasks that may be possible
in general but not in a theory where the effective number of
outcomes is restricted. Namely, perfect discrimination of $n$
states requires that $\sum_x \lmax(\A_x) \geq n$. Hence, by Prop.
\ref{prop:eigen-gen}(b) an effectively $n$-tomic meter can
discriminate at most $n$ states.

As a consequence of Prop. \ref{prop:lambda1}, there exists a
dichotomic meter $\A$ with $\sum_x \lmax(\A_x) = 2$. Therefore,
the bound for effectively dichotomic meters in Prop.
\ref{prop:eigen-gen}(b) cannot be improved without additional
assumptions. The following statement has specific assumptions and
for that reason gives a tighter bound. In Example \ref{ex:ud}
below we show that this result has interesting implications.

\begin{proposition} \label{prop:not-2-eff}
Let $\A$ be an $n$-outcome meter such that $\A_1, \ldots, \A_m$
are  indecomposable effects for some $m\leq n$ and for all $i,j
\in \{1, \ldots, m \}$ such that $i\neq j$ we have
\begin{enumerate}
\item[i)]  $\A_j \neq t \A_i$ for all $t>0$,
\item[ii)] $t_i\A_i + t_j \A_j \neq u$ for all $t_i,t_j>0$.
\end{enumerate}
If $\sum_{k=1}^m \lmax(\A_k) >1$, then $\A$ is not effectively dichotomic.
\end{proposition}
\begin{proof}
Suppose $\A$ is effectively dichotomic so that there exist dichotomic meters $\{\B^{(i)}\}_{i=1}^l$, a probability distribution $(p_i)_{i=1}^l$ and post-processings $\nu^{(i)}: \{+, -\} \to \{1 ,\ldots, n\}$  for all $i=1, \ldots,l$ such that
\begin{align*}
\A_j = \sum_{i=1}^l p_i \left( \nu^{(i)}_{+j} \B^{(i)}_+ + \nu^{(i)}_{-j} \B^{(i)}_- \right)
\end{align*}
for all $j \in \{1, \ldots, n \}$, where we may assume that $p_i \neq 0$ for all $i=1, \ldots,l$. By the assumption, $\A_j$ is indecomposable for all $j \in \{1, \ldots, m\}$. Thus, for each $j$, there exists index sets $I^{(j)}_\pm := \{ 1\leq i \leq l \, | \, \nu^{(i)}_{\pm j} \neq 0 \}$ such that $\B^{(i)}_+ = \alpha^{(j)}_i \A_j$ and $\B^{(k)}_- = \beta^{(j)}_k \A_j$ for some $\alpha^{(j)}_i, \beta^{(j)}_k \in (0,1]$ for all $i \in I^{(j)}_+$ and $k \in I^{(j)}_-$.

First of all, we note that $I^{(j)}_+ \cap I^{(j)}_- = \emptyset$ for all $j \in \{1, \ldots,m\}$ because otherwise $\A_j$ would be proportional to $u$ due to the normalisation of $\B^{(i)}$'s. Secondly, from $i)$ it follows that $I^{(j)}_+ \cap I^{(k)}_+ = I^{(j)}_- \cap I^{(k)}_- = \emptyset$ for all $k,j \in \{1, \ldots,m\}$ such that $j\neq k$. Thirdly, from $ii)$ it follows that $I^{(j)}_+ \cap I^{(k)}_- = \emptyset$ for all $k,j \in \{1, \ldots,m\}$ such that $j\neq k$. Thus, the sets $\{I^{(j)}_\pm \}_{j=1}^m$ form a partition of their union $I:= \bigcup_{j=1}^m \left( I^{(j)}_+ \cup I^{(j)}_- \right) \subseteq \{1, \ldots, l\}$.

We can now write
\begin{equation}
\A_j = \sum_{i \in I^{(j)}_+} p_i \nu^{(i)}_{+j} \B^{(i)}_+ + \sum_{k \in I^{(j)}_+} p_k \nu^{(k)}_{+j} \B^{(k)}_+ \leq \left( \sum_{i \in I^{(j)}_+} p_i + \sum_{k \in I^{(j)}_+} p_k \right) u.
\end{equation}
From the above expression and the properties of the index sets it follows that
\begin{align*}
\sum_{j=1}^m \lmax(\A_j) \leq \sum_{j=1}^m \left( \sum_{i \in I^{(j)}_+} p_i +  \sum_{k \in I^{(j)}_-} p_k \right) = \sum_{i \in I} p_i \leq \sum_{i=1}^l p_i =1.
\end{align*}
\end{proof}

\begin{example}\label{ex:ud}
(\emph{Unambiguous discrimination of two qubit states}) Let
$\varrho_1 = \ket{\psi_1}\bra{\psi_1}$ and $\varrho_2 =
\ket{\psi_2}\bra{\psi_2}$ be two pure qubit states with \textit{a
priori} probabilities $p_1 = p_2 = \frac{1}{2}$. The unambiguous
discrimination of these states involves a 3-outcome POVM with
effects $\A_1,\A_2,\A_{?}$ such that observation of the outcome 1
(2) guarantees that the input state was $\varrho_1$ ($\varrho_2$).
This implies
$$
{\rm tr}[\varrho_1 A_2] = {\rm tr}[\varrho_2 A_1] =
0
$$
and hence
$$
\A_1 = q_1 (\id - \ket{\psi_2}\bra{\psi_2})\, , \quad \A_2 = q_2 (\id
- \ket{\psi_1}\bra{\psi_1})
$$
for some $q_1,q_2 > 0$ such that $\A_{?} = \id - \A_1 - \A_2$ is a valid effect, i.e., $\A_{?} \geq O$.
Suppose $\A$ is effectively dichotomic. Then by Prop. \ref{prop:not-2-eff} we have $q_1 + q_2 \leq 1$ and the success probability is
\begin{align}\label{eq:usd-bound}
p_{\rm success} & = \half \tr{\varrho_1 A_1} + \half \tr{\varrho_2 A_2} \nonumber \\
& = \frac{q_1 + q_2}{2} \left( 1 - | \langle
\psi_1 | \psi_2 \rangle |^2 \right) \leq \frac{1}{2} \left( 1 - |
\langle \psi_1 | \psi_2 \rangle |^2 \right) \, .
\end{align}
However, it is known \cite{QDET76} that the optimal success probability without any limitations is $1-\mo{\ip{\psi_1}{\psi_2}}$.
This is strictly higher than the bound in \eqref{eq:usd-bound} whenever $\varrho_1$ and $\varrho_2$ are two different states.
We conclude that the restriction to dichotomic meters decreases the optimal success probability in unambiguous discrimination.
\end{example}

\section{Restriction class (R3), noise and compatibility }\label{sec:R3}

In this section we present some examples of restrictions that arise quite naturally and belong to the class (R3).

\subsection{Compatibility restriction}

We recall that two meters $\A$ and $\B$ are compatible if they can be simulated with a single meter $\C$, i.e., $\{\A,\B\} \subset \simu{\C}$. Let us a fix a meter $\A$ and consider all meters that are compatible with $\A$; we denote this set by $C(\A)$.  The conditions for $C(\A)\neq\meter$ have been characterized in \cite{HeLePl19}.
In the following we assume that $C(\A)\neq\meter$ and choose $\rmeter=C(\A)$.

\begin{proposition} \label{prop:compatibility}
$\simu{C(\A)}=C(\A)$.
\end{proposition}
\begin{proof}
To see this, take $\D \in \simu{C(\A)}$ so that there exist meters
$\{\B^{(i)}\}_i \subset C(\A)$ such that $\D = \sum_{i=1}^n p_i
(\nu^{(i)} \circ \B^{(i)})$ for some probability distribution
$(p_i)_{i=1}^n$ and post-processings $\nu^{(i)}: \Omega_{\B^{(i)}}
\to \Omega_\D$ for all $i \in \{1, \ldots,n\}$. If we define a new
meter $\tilde{\B}$ as $\tilde{\B}_{(i,x)} = p_i \B^{(i)}_x$ for
all $i \in \{1, \ldots, n\}$ and $x \in \Omega_{\B^{(i)}}$ (where
we can take $\Omega_{\B^{(i)}} = \Omega_{\B^{(j)}}=: \Omega_\B$
for all $i,j$), we see that
\begin{equation*}
\D_y = \sum_{i,x} p_i \nu^{(i)}_{xy} \B^{(i)}_x = (\nu \circ \tilde{\B})_y
\end{equation*}
for all $y \in \Omega_\D$, where we have defined a post-processing $\nu: \{1, \ldots,n\} \times \Omega_\B \to \Omega_\D$ by $\nu_{(i,x)y}= \nu^{(i)}_{xy}$ for all $i \in \{1, \ldots,n\}$, $x \in \Omega_\B$ and $y \in \Omega_\D$. Thus, $\D \in \simu{\tilde{\B}}$.

Since $\B^{(i)} \in C(\A)$, for any $i \in \{1, \ldots, n\}$ there
exists $\C^{(i)} \in \meter$ such that $\{\A, \B^{(i)}\} \subset
\simu{\C^{(i)}}$. Similarly to $\tilde{\B}$, we can define
$\tilde{\C}$ by setting $\tilde{\C}_{(i,z)} = p_i \C^{(i)}_z$ for
all $i \in \{1, \ldots, n\}$ and $z \in \Omega_\C$. Since $\A \in
\simu{\C^{(i)}}$, there exists a post-processing $\mu^{(i)}:
\Omega_\C \to \Omega_\A$ for all $i \in \{1, \ldots, n\}$ such
that $\A = \mu^{(i)} \circ \C^{(i)}$ so that
\begin{equation*}
\A_k = \sum_i p_i \A_k = \sum_{i} p_i (\mu^{(i)} \circ \C^{(i)})_k = \sum_{i,z} p_i \mu^{(i)}_{zk} \C^{(i)}_z = (\mu \circ \tilde{\C})_k
\end{equation*}
for all $k \in \Omega_\A$, where we have defined another post-processing $\mu: \{1, \ldots,n\} \times \Omega_\C \to \Omega_\A$ by $\mu_{(i,z)k}= \mu^{(i)}_{zk}$ for all $i \in \{1, \ldots,n\}$, $z \in \Omega_\C$ and $k \in \Omega_\A$. Thus, $\A \in \simu{\tilde{\C}}$.

On the other hand, since $\B^{(i)} \in \simu{\C^{(i)}}$ for all $i \in \{1, \ldots, n\}$, there exists post-processings $\kappa^{(i)}: \Omega_\C \to \Omega_\B$ such that $\B^{(i)} = \kappa^{(i)} \circ \C^{(i)}$ so that
\begin{equation*}
\tilde{\B}_{(i,x)} = p_i \B^{(i)}_x = \sum_z \kappa^{(i)}_{zx} p_i \C^{(i)}_z = ( \kappa \circ \tilde{\C})_{(i,x)}
\end{equation*}
for all $(i,x) \in \{1, \ldots, n\} \times \Omega_\B$, where we
have defined yet another post-processing $\kappa: \{1, \ldots,n\}
\times \Omega_\C \to \{1, \ldots, n\} \times \Omega_\B$ by
$\kappa_{(j,z)(i,x)}= \delta_{ij} \kappa^{(j)}_{zx}$ for all $i,j
\in \{1, \ldots,n\}$, $z \in \Omega_\C$ and $x \in \Omega_\B$.
Hence, $\tilde{\B} \in \simu{\tilde{\C}}$ and $\D \in
\simu{\tilde{\C}}$.

To conclude, we have shown that $\{\A, \D\} \subset
\simu{\tilde{\C}}$, i.e., $\A$ and any $\D \in \simu{C(\A)}$ are
compatible, therefore $\simu{C(\A)} \subset C(\A)$ and $C(\A)$ is
simulation closed.
\end{proof}

Interestingly, in quantum theory the restriction $C(\A)$ can be
either (R1) or (R3), depending on $\A$. Firstly, if $\A$ is a
sharp quantum meter, i.e., every $\A_x$ is a projection operator,
then a meter $\B$ is compatible with $\A$ if and only if
$[\A_x,\B_y]=O$ for all outcomes $x,y$ \cite{PSAQT82}. The
restriction $C(\A)$ is then of the type (R1) as
$C(\A)=\meter_{\reffect}$, where
\begin{equation*}
\reffect = \{E \in \eh : [E,\A_x]=O \ \forall x \} \, .
\end{equation*}
Secondly, to see that the restriction $C(\A)$ can be of the type
(R3) we recall the result in \cite{ReReWo13}, which demostrates
the existence of quantum meters $\A$ and $\B$ in $\complex^3$ such
that $\A$ is a dichotomic, $\B$ is trichotomic, and they are
coexistent but not compatible. The coexistence of $\A$ and $\B$
means that all coarse-grainings of $\B$ into dichotomic meters are
compatible with $\A$. The union of the ranges of all dichotomic
coarse-grainings of $\B$ is the same as the range of $\B$. This
result implies that there is no $\reffect$ such that $C(\A)=
\meter_{\reffect}$. Finally, to see that $C(\A)$ cannot be of the
type (R2), we observe that $\effect_{C(\A)}=\effect$ implies that
$\A$ is compatible with all dichotomic meters. If this is the
case, every $\A_x$ commutes with all projection operators and
hence $\A_x$ is a multiple of the identity operator $\id$. But
then $C(\A)=\meter$ and we do not have a restriction at all, which
is a contradiction in general so $\effect_{C(\A)} \neq \effect$.

\subsection{Noise restriction on meters}
Let us denote by $\mathcal{P}(\Omega)$ the set of probability distributions on a (finite) set $\Omega$. Let us fix $t \in [0,1]$ and define a restriction $\rmeter_t$ on meters as
\begin{equation} \label{eq:noisy-meters}
\rmeter_t = \{ t \B +(1-t) p u \, | \, \B \in \meter, \ p \in
\mathcal{P}(\Omega_\B) \} .
\end{equation}
Clearly, if $t =1$, we have $\rmeter_1 = \meter$, and if $t= 0$ we have $\rmeter_0 = \trivial$, where $\trivial$ is the set of trivial meters. Thus, we can interpret the parameter $t$ as noise on the meters so that the smaller $t$ gets, the noisier the meters in $\rmeter_t$ become.

Let now $t \in (0,1)$. We will show that then $\rmeter_t$ is a
restriction of type (R3). First of all, we see that $\rmeter_t$ is
simulation closed: If we take $\A \in \simu{\rmeter_t}$ so that
$\A = \sum_i p_i (\nu^{(i)} \circ \B^{(i)} )$ for some meters
$\{\B^{(i)} = t \C^{(i)} +(1-t) q^{(i)} u\}_i \subset \rmeter_t$,
some post-processings $\nu^{(i)}: \Omega_\B \to \Omega_\A$, and
some probability distribution $(p_i)_i$, then
\begin{equation*}
\A_y = \sum_{i,x} p_i \nu^{(i)}_{xy} \B^{(i)}_x = \sum_{i,x} p_i \nu^{(i)}_{xy} [t \C^{(i)}_x +(1-t) q^{(i)}_x u] = t \C_y +(1-t) q_y u,
\end{equation*}
where we have defined a new meter $\C = \sum_i p_i (\nu^{(i)} \circ \C^{(i)}) \in \meter$ and a new probability distribution $q \in \mathcal{P}(\Omega_\A)$ by setting $q_y = \sum_{i,x} p_i \nu^{(i)}_{xy} q^{(i)}_x$. Thus, $\A = t\C +(1-t) q u \in \rmeter_t$ so that $\rmeter_t$ is simulation closed.

Next we see that
\begin{equation}\label{eq:noisy-effects}
\effect_{\rmeter_t} = \{ t e + (1-t) r u \, | \, e \in \effect, \ r \in [0,1] \}
\end{equation}
so that $\effect_{\rmeter_t} \subsetneq \effect$ since $t \neq 1$. Thus, $\rmeter_t$ is not of type (R2). What remains to show is that $\rmeter_t  \neq \meter_{\reffect}$ for all effect restrictions $\reffect \subset \effect$.

Our first observation is that if a restriction on meters $\rmeter$ is induced by some effect restriction $\reffect$, i.e., $\rmeter= \meter_{\reffect}$, then the effect restriction $\reffect$ is unique and is given by the induced effects $\effect_{\rmeter} = \effect_{\meter_{\reffect}}$.

\begin{proposition}\label{prop:effect_restriction}
For a restriction $\meter_{\reffect}$ induced by an effect
restriction $\reffect \subset \effect$ satisfying the
consistency conditions (E1) and (E2) we have that
$\effect_{\meter_{\reffect}}= \reffect$.
\end{proposition}
\begin{proof}
Let us take $e \in \effect_{\meter_{\reffect}}$ so that there
exists $\A \in \meter_{\reffect}$ such that $e \in \ran(\A)$. From
the definition of $\meter_{\reffect}$ it follows that $e \in
\ran(\A) \subset \reffect$. Thus, $\effect_{\meter_{\reffect}}
\subseteq \reffect$.

For the other direction let us take $f \in \reffect$. As it was
stated earlier, any dichotomic meter $\F$ with effects $f$ and
$u-f$ must be in $\meter_{\reffect}$ so that from the definition
of $\effect_{\meter_{\reffect}}$ we see that $f \in
\effect_{\meter_{\reffect}}$. Thus, $\reffect \subseteq
\effect_{\meter_{\reffect}}$. Combining
$\effect_{\meter_{\reffect}} \subseteq \reffect$ and $\reffect
\subseteq \effect_{\meter_{\reffect}}$, we have the claim.
\end{proof}

Thus, for $\rmeter_t$, the previous result implies that if
$\rmeter_t = \meter_{\reffect}$ for some $\reffect \subset
\effect$, then $\reffect = \effect_{\rmeter_t}$. First of all, one
can readily see that $\effect_{\rmeter_t}$ is convex and hence, by
Theorem \ref{thm:R1}, $\meter_{\effect_{\rmeter_t}}$ is simulation
closed as it should be. We will proceed by constructing a meter
$\A \in \meter_{\reffect}$ such that $\A \notin \rmeter_t$.

To see when a given meter is in $\rmeter_t$, we give a convenient
characterization for $\rmeter_t$ in terms of the \emph{noise
content} of a meter \cite{fhl-2017}. The noise content $w(\B;
\noise)$ of a meter $\B \in \meter$ with respect to a noise set
$\noise \subset \meter$ is defined as
\begin{equation*}
w(\B; \noise) = \sup \{ 0 \leq \lambda \leq 1 \, | \, \exists \C
\in \meter, \N \in \noise: \ \B = \lambda \N +(1-\lambda) \C \}.
\end{equation*}
The noise content $w(\B; \noise)$ thus characterizes how much of
$\B$ is in $\noise$ with respect to the convex structure of
meters. When $\noise$ is chosen to represent some noise in the
meters, the noise content can be interpreted as the amount of the
intrinsic noise that is present in the meter (contrary to the
external noise that is typically added to a meter). A typical
choice for $\noise$ is to set $\noise = \trivial$, the set of
trivial meters. In this case it can be shown that
\begin{equation}\label{eq:noise-content}
w(\B; \trivial) = \sum_{x \in \Omega_\B} \lmin(\B_x).
\end{equation}

We can now give the following characterization for $\rmeter_t$:
\begin{lemma}\label{lemma:noise-restriction}
Meter $\B \in \rmeter_t$ if and only if $w(\B; \trivial) \geq 1-t$.
\end{lemma}
\begin{proof}
Let first $\B \in \rmeter_t$ so that there exists $\C \in \meter$ and $p \in \mathcal{P}(\Omega_\B)$ such that $\B = t \C +(1-t) p u= t \C +(1-t) \T$, where we have defined a trivial meter $\T \in \trivial$ by $\T_x =p_x u$ for all $x \in \Omega_\B$. From the definition of the noise content we see that $w(\B; \trivial) \geq 1-t$.

Let then $w(\B; \trivial) \geq 1-t$. Since we have the noise set
$\noise= \trivial$, by Eq. \eqref{eq:noise-content} the supremum
in the definition of the noise content is attained so there exist
$\D  \in \meter$ and $\T \in \trivial$ such that $\B = w(\B;
\trivial) \T +(1-w(\B;\trivial)) \D$. We have
\begin{align*}
\B &= w(\B; \trivial) \T + (1-w(\B; \trivial) ) \D \\
&= (1-t+t+w(\B; \trivial)-1) \T +(1- w(\B; \trivial)) \D \\
&= (1-t) \T + t \left[ \dfrac{t+w(\B;\trivial)-1}{t} \T + \dfrac{1- w(\B; \trivial)}{t} \D \right] \\
&= (1-t) \T + t \tilde{\D} \in \rmeter_t,
\end{align*}
where $\tilde{\D}= \frac{t+w(\B;\trivial)-1}{t} \T + \frac{1- w(\B; \trivial)}{t} \D \in \meter$ is a convex mixture of $\T$ and $\D$.
\end{proof}

Now we are ready to prove that $\rmeter_t$ is a restriction of type (R3) for all $t \in (0,1)$.
\begin{proposition} \label{prop:noisy-meters}
Let $t \in (0,1)$. Then $\rmeter_t \neq \meter_{\reffect}$ for any
$\reffect \subset \effect$.
\end{proposition}
\begin{proof}
Let us suppose that $\rmeter_t = \meter_{\reffect}$ for some
effect restriction $\reffect \subset \effect$. As it was mentioned
earlier, by Prop. \ref{prop:effect_restriction} we then have that
$\reffect= \effect_{\meter_{\reffect}} = \effect_{\rmeter_t}$. We
will construct a meter $\A \in \meter_{\reffect}$ such that $\A
\notin \rmeter_t$, which will then be a contradiction.

Let us start the construction of $\A$, by constructing another meter $\B$ with $w(\B;\trivial) =0$ and $\max_{x \in \Omega_\B} \lmax(\B_x) \in [t,1)$. We will then use $\B$ to construct $\A$ and use Lemma \ref{lemma:noise-restriction} together with the previously listed properties of $\B$ to show that $\A \notin \rmeter_t$.

Let us fix an extreme indecomposable effect $e \in \effect(\state)$. If we set $e_1
= e$ and decompose $u-e$ into indecomposable effects $u-e =
\sum_{i=2}^n e_i$ for some $n \in \nat$, we can define a meter
$\tilde{\B}$ as $\tilde{\B}_i = e_i$ for all $i \in \{1,
\ldots,n\}$ such that all of the effects of $\tilde{\B}$ are
indecomposable. Since indecomposable effects lie on the boundary
of the positive cone of the effect space, we have that
$\lmin(\tilde{\B}_i) =0$ for all $i \in \{1, \ldots,n\}$ so that
by Eq. \eqref{eq:noise-content} we have that $w(\tilde{\B};
\trivial) =0$.

Let us relabel the outcomes of $\tilde{\B}$ in such a way that
$\lmax(\tilde{\B}_i) = 1$ for all $i \in \{1, \ldots, m\}$ for
some $m \leq n$. We recall that since $e$ is an extreme effect we
know that $\lmax(e) =1$, and since $e \in
\mathrm{ran}(\tilde{\B})$ we must have $m\geq 1$. We take $q \in
[t,1)$ and define a new meter $\B$ with effects
\begin{equation}
\B_i =
\begin{cases}
q \tilde{\B}_i, & i\in \{1, \ldots,m\} \\
\tilde{\B}_i, & i \in \{m+1, \ldots, n\} \\
(1-q) \tilde{\B}_{i-n}, & i \in \{n+1, \ldots, n+m\}
\end{cases}.
\end{equation}
By construction we have that $w(\B;\trivial) =0$ and $l_\B:=\max_{x \in \Omega_\B} \lmax(\B_x) \in [t,1)$.

Let us now take numbers $\{r_i\}_{i=1}^{n+m} \subset [0,1]$ such
that $r:= \sum_{i=1}^{n+m} r_i \in \left[
\frac{l_\B-t}{(1-t)l_\B}, 1 \right)$ and define a new meter $\A$
by $\A_i = t a_i +(1-t)r_i u$, where we have defined $a_i =
\frac{1-(1-t)r}{t} \B_i$ for all $i \in \{1, \ldots, n+m\}$. Once
we show that $\A$ is well-defined and $\A \in
\meter_{\effect_{\meter_t}}$, we can use Eq.
\eqref{eq:noise-content} to see that $w(\A; \trivial) < 1-t$ so
that by Lemma \ref{lemma:noise-restriction} we have $\A \notin
\rmeter_t$ which completes the proof.

In order to show that $\A$ is well-defined we need to show that
$\A$ is a meter and that we can choose $\{r_i\}_i$ like we wanted.
The problematic parts in the definition of the sequence
$\{r_i\}_i$ are that we might have that
$\frac{l_\B-t}{(1-t)l_\B}<0$, which might lead to $r<0$, or
$\frac{l_\B-t}{(1-t)l_\B}\geq 1$, which would leave the interval
$\left[ \frac{l_\B-t}{(1-t)l_\B}, 1 \right)$ empty. However, from
the definition of $\B$ we see that $l_\B = \max_{x \in \Omega_\B}
\lmax(\B_x) \geq t$ so that $\frac{l_\B-t}{(1-t)l_\B} \geq 0$, and
since $l_\B <1$ it is easy to see that $\frac{l_\B-t}{(1-t)l_\B} <
1$. Thus, we can choose the sequence $\{r_i\}_i$ like we wanted.

In order to show that $\A \in \meter_{\effect_{\meter_t}}$ we need
to show that $a_i \in \effect(\state)$ for all $i \in \{1,
\ldots,n+m\}$ and that $\sum_i \A_i = u$. Since
$r<1<\frac{1}{1-t}$ we see that $\frac{1-(1-t)r}{t} >0$ so that
$a_i = \frac{1-(1-t)r}{t} \B_i\geq o$ for all $i \in \{1, \ldots,
n+m\}$. On the other hand, we have $a_i \leq u$ if and only if $
\frac{1-(1-t)r}{t} \lmax(\B_i) \leq 1$ which is equivalent to $r
\geq \frac{\lmax(\B_i) -t}{(1-t) \lmax(\B_i)}$. Since $r \geq
\frac{l_\B -t}{(1-t) l_\B} \geq \frac{\lmax(\B_i) -t}{(1-t)
\lmax(\B_i)}$, it follows that $a_i \leq u$ for all $i \in \{1,
\ldots,n+m\}$. Thus, $a_i \in \effect(\state)$ so that $\A_i \in
\effect_{\rmeter_t}$ for all $i \in \{1, \ldots,n+m\}$.
Furthermore, we see that
\begin{equation*}
\sum_{i=1}^{n+m} \A_i = (1-(1-t)r)  \left( \sum_{i=1}^{n+m} \B_i \right)   + (1-t)r u = (1-(1-t)r)u +(1-t)r u = u.
\end{equation*}
Hence, $\A \in \meter_{\effect_{\rmeter_t}}$.

For the noise content of $\A$, we see that
\begin{align*}
w(\A;\trivial) &= \sum_{i=1}^{n+m} \lmin(\A_i) = t  \left( \sum_{i=1}^{n+m} \lmin(a_i)\right) + (1-t)r  \\
&= (1-(1-t)r) \left(  \sum_{i=1}^{n+m} \lmin(\B_i) \right) +(1-t)r \\
&= (1-t)r < 1-t,
\end{align*}
which by Lemma \ref{lemma:noise-restriction} shows that $\A \notin \rmeter_t$.
\end{proof}

Thus, we have just demonstrated that if the noise is introduced at
the level of meters as in Eq. \eqref{eq:noisy-meters}, the induced
restriction cannot be reproduced by considering noise on effects
alone. However, one can of course start with Eq.
\eqref{eq:noisy-effects} and use it as a restriction on its own so
that we will naturally arrive at (R1) type of restriction instead.
The next example will illustrate this point in quantum theory.

\begin{figure}[t]
\begin{center}
\includegraphics[scale=1]{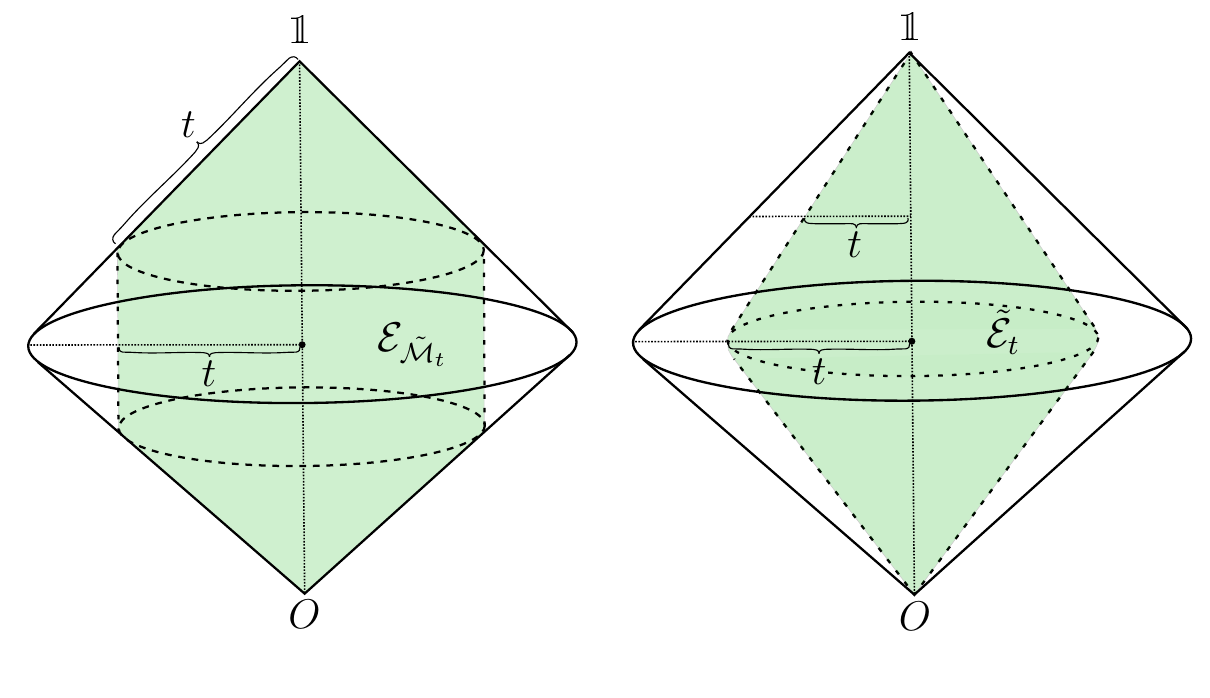}
\end{center}
\caption{\label{fig:qubit-restrictions} The effect restrictions given by Eq. \eqref{eq:noisy-effects} (on left) and \eqref{eq:depolarizing-noise} (on right) on a three-dimensional cross section of the qubit effect space. They are both convex subsets of effects that satisfy the consistency conditions so that they induce valid simulation closed restrictions on meters.}
\end{figure}

\begin{example}[\emph{Depolarizing noise in quantum theory}]\label{ex:depolarize}
In quantum theory, the standard depolarizing channel $\Phi_t: \lh \to \lh$ on a $d$-dimensional Hilbert space $\hi$ is defined as
\begin{equation}
\Phi_t(\varrho) = t \varrho +(1-t) \tr{\varrho} \frac{\id}{d}
\end{equation}
for all $\varrho \in \lh$ with some noise parameter $t \in [0,1]$.
In the Heisenberg picture, the depolarizing noise can be
alternatively ascribed to the meters, which results in the
restricted set of effects
\begin{equation}\label{eq:depolarizing-noise}
\tilde{\effect}_t = \Phi^*_t(\effect(\state)) = \left\lbrace t E
+(1-t)\frac{\tr{E}}{d}\id \, : \, E \in \effect(\hi)
\right\rbrace,
\end{equation}
where $\Phi^*_t$ is dual to $\Phi_t$. Clearly $\tilde{\effect}_t
\subseteq \effect(\hi)$ for all $t \in (0,1]$ and the equality
holds only if $t=1$. For $t \in (0,1]$ it is straigthforward to
verify that $\Phi^*_t$ is an affine isomorphism between
$\effect(\hi)$ and $\tilde{\effect}_t$ so that by Prop.
\ref{prop:restricting-isomorphism} we can deduce that $\reffect_t$
is a restriction of type (R1) that does not form a convex
subalgebra of $\effect(\state)$. For $t=0$ we have that
$\reffect_0 = \mathrm{span}_{[0,1]}\{\id\}$, which is a trivial
convex subalgebra of every effect algebra.

However, if we consider a class of general (shifted) depolarizing
channels $\Psi_{t,\xi}(\varrho) = t \varrho + (1-t) {\rm
tr}[\varrho] \xi$ with a general state $\xi$ instead of the
maximally mixed state $\id/d$, then a wider class of effects is
achievable. This describes a physically relevant situation when
the considered qubit is coupled to a two-level
fluctuator~\cite{PaGaFaAl14}. The dual map then reads
$\Psi^*_{t,\xi}(E) = t E+ (1-t) \tr{E \xi} \id$ so that in the
case of quantum theory it can be confirmed that $\{
\Psi^*_{t,\xi}(\effect(\hi))\, | \, \xi \in \sh\} =
\effect_{\rmeter_t}$, i.e., we get all the effects provided by Eq.
\eqref{eq:noisy-effects}. Clearly, $\effect_{\rmeter_t} \neq
\reffect_t$. The effect restrictions $\effect_{\rmeter_t}$ and
$\reffect_t$ are depicted in Fig. \ref{fig:qubit-restrictions}.
\end{example}

\section{Discussion and conclusions}\label{sec:conclusions}

Our primary goal in this paper was to establish a natural
criterion that any operational restriction is to satisfy and to
classify such restrictions. Given a set of meters one can always
randomly switch among the meters and classically postprocess their
measurement outcomes. As a result, one readily gets a simulation
closure of the original set of meters. Equipped with the natural
operational requirement of simulation closedness, we have divided
all operational restrictions into three classes.

Class (R1) describes such restrictions that originate from the
truncation of the set of effects. We have characterized such sets of effects in Theorem \ref{thm:R1}. We have demonstrated that a restriction to any
convex subalgebra of the set of all effects induces a proper
operational restriction of class (R1). Further in
Proposition~\ref{prop:restricting-isomorphism}
 we have proved that there exist operational restrictions of class
(R1) that do not reduce to convex effect subalgebras.
Proposition~\ref{prop:effect_restriction} clarifies that the
effect restriction is unique if the consistency conditions (E1)
and (E2) are satisfied.

Surprisingly enough, there exist operational restrictions of class
(R2) on meters such that every effect within the no-restriction
hypothesis is accessible, however, the set of meters is severely
truncated. The most prominent example is effectively dichotomic
meters, a simulation closure of dichotomic meters. Moreover, any
restriction of class (R2) must contain effectively dichotomic
meters as a subset
(Proposition~\ref{proposition-2-eff-subset-R2}).

It is worth mentioning that effectively dichotomic meters
naturally emerge in conventional experiments with polarized
photons and superconducting qubits, and therefore are of great
practical interest. Despite the fact that restrictions of class
(R2) seem quite innocent as compared to the restrictions of class
(R1), they do impose some strong physical limitations. In
Example~\ref{ex:ud} we have demonstrated that the success
probability of unambiguous discrimination of nonorthogonal pure
qubit states with effectively dichotomic meters is strictly less
than that with trichotomic meters. From a wider viewpoint or
resource theories~\cite{CoFrSp16,ChGo19}, Example~\ref{ex:ud}
opens an avenue for the study of the resource theory of
$n$-tomicity. Within such a resource theory, $n$-outcome meters
are free and any simulation scheme for meters is a free operation.
A meter that is not $n$-tomic may represent a resource for some
task (as a trichotomic meter in the unambiguous discrimination in
Example~\ref{ex:ud}). As a byproduct of this research direction,
we have also derived some sufficient and (separately) necessary
conditions for effectively $n$-tomic observables
(Propositions~\ref{prop:lambda1}--\ref{prop:not-2-eff}).

Finally, we have demonstrated that there are restrictions that
arise rather naturally but belong to neither (R1) nor (R2). We
have shown that such restrictions can emerge when one considers
meters compatible with a given meter
(Proposition~\ref{prop:compatibility}) or when one tries to
account for noise in the meters
(Proposition~\ref{prop:noisy-meters} and
Example~\ref{ex:depolarize}). We believe that the operational
restrictions of type (R3) can be further analyzed in subsequent
works.

\section{Acknowledgements}
S.N.F. and T.H. acknowledge the support of the Academy of Finland
for mobility grants to visit University of Turku and Moscow
Institute of Physics and Technology, respectively. T.H. and L.L.
acknowledge the support from the Academy of Finland via the Centre
of Excellence program (Grant No. 312058) as well as Grant No.
287750. L.L. acknowledges financial support from University of
Turku Graduate School (UTUGS).

\vspace*{1cm}

\end{document}